%% file: BS-RS-FIR-C-A.tex
\RequirePackage{amsmath}
\documentclass[english,twocolumn]{article}

\usepackage[big]{dgruyter}
\setlocalecaption{ngerman}{abstract}{Zusammenfassung}

\makeatletter 
\newcommand\mynobreakpar{\par\nobreak\@afterheading} 
\makeatother

\usepackage{wrapfig}

\input{mypackages}

\renewcommand{\tilde}[1]{\widetilde{#1}}

\begin{document}

  \author*[1]{Sebastian Schlor}
  \author[2]{Frank Allgöwer}
  \runningauthor{Schlor and Allgöwer}
  \affil[1]{S.\ Schlor and F.\ Allgöwer are with the University of Stuttgart, Institute for Systems Theory and Automatic Control, Germany.  {\tt\small \{schlor, allgower\}@ist.uni-stuttgart.de}.}
  \title{Comparison and performance analysis of dynamic encrypted control approaches}
  \runningtitle{Comparison and performance analysis of dynamic encrypted control approaches}
  \abstract{
  	Encrypted controllers using homomorphic encryption have proven to guarantee the privacy of measurement and control signals, as well as system and controller parameters, while regulating the system as intended. However, encrypting dynamic controllers has remained a challenge due to growing noise and overflow issues in the encoding. 
  	In this paper, we review recent approaches to dynamic encrypted control, such as bootstrapping, periodic resets of the controller state, integer reformulations, and FIR controllers, and equip them with a stability and performance analysis to evaluate their suitability. 
  	We complement the analysis with a numerical performance comparison on a benchmark system.
  }
  \startpage{1}
  \aop

\maketitle

\section{Introduction}

As modern production systems become more interconnected, with remote monitoring access and cloud-based networked-control approaches, privacy and security aspects become increasingly important for the operators.
In cases where data is stored and processed on external servers, encrypted computations can provide privacy from data breaches while providing functionality by operating with private data.
Thanks to modern homomorphic encryption (HE) schemes such as the Cheon-Kim-Kim-Song (CKKS) cryptosystem~\cite{Cheon2017}, computations on encrypted data can be done without giving the cloud provider access to the decryption key.
Since the first work on encrypted control~\cite{Kogiso2015}, solutions to many relevant problems such as encrypted static state feedback, encrypted data-driven control~\cite{Alexandru}, and encrypted PID tuning~\cite{Schluter2022a} have been proposed. For a detailed summary of the recent advances, see~\cite{Schlueter2023}.
Dynamic encrypted controllers, however, have posed a challenge because typical HE schemes only allow for a limited number of multiplications, whereas in dynamic controllers, the state is updated in every time step, typically involving a multiplication. Thus, by default, these encrypted dynamical systems can only be updated for a finite number of times.
The cause for this limited number of possible multiplications is twofold: If the cryptosystem relies on added noise for security, repeated multiplications lead to noise growth, and eventually, failing decryption. Further, fixed-point number representations require a scaling factor, which is amplified by repeated multiplication, leading to overflows and incorrectly represented numbers.

\subsection{Contribution}
In the recent years, several different solution approaches for dynamic encrypted control have been presented, involving very different concepts from control theory and cryptography.
Here, we summarize these approaches, present a performance analysis for the closed loop, and discuss their practical implementation.
In particular, we make the following contributions:
\begin{itemize}
		\item We give an overview of the approaches to dynamic encrypted control and review the recent literature.
		\item We provide a general stability and performance analysis for the presented techniques.
		\item We compare and unify the analysis of reset, finite-impulse-response (FIR) and bootstrapping controllers for a better understanding of the underlying errors made.
		\item We present a new integer reformulation by feedback in observer form.
\end{itemize}%
With this, we provide a basis for further research in the area of encrypted dynamic control.

\subsection{Outline}
The paper is structured as follows. 
In Section~\ref{sec:problem}, we precisely formulate the considered problem. 
Section~\ref{sec:review} presents the currently available approaches to dynamic encrypted control and their advantages and disadvantages. 
In Section~\ref{sec:analysis}, a stability and performance analysis of the reviewed approaches is provided.
We present a numerical comparison in Section~\ref{sec:numerics} and conclude the paper in
Section~\ref{sec:Summary}.

\section{Problem formulation}\label{sec:problem}

In this section, we describe the considered dynamical system, the dynamic controller and their feedback interconnection.
As a baseline we take the unencrypted controller to which we compare the encrypted control approaches later.

\subsection{System description}

We consider a discrete-time, linear, time-invariant system, represented by
\begin{equation}
	\begin{aligned}
	x(t+1) &= A x(t) + Bu(t) + B_1 w_{p_1}(t)\\
	y(t) &= C x(t) + F_1 w_{p_1}(t)\\
	z_p(t) &= C_1x(t) + Eu(t) + D_1w_{p_1}(t)
\end{aligned}\label{eq:plant}
\end{equation}
with time-index $t\in\N$, initial condition $x(0)=x_0$, control input $u$, performance input $w_{p_1}$, measurement output $y$, and performance output $z_p$.
For this system, we consider a predefined dynamic output feedback controller 
\begin{equation}
	\begin{aligned}
	x_c(t+1) &= A_c x_c(t) + B_c y(t) + B_2w_{p_2}(t)\\
	u(t) &= C_c x_c(t) + D_c y(t) + F_2 w_{p_2}(t)
\end{aligned}\label{eq:controller}
\end{equation}
with initial state $x_c(0)=x_{c,0}$, performance input $w_{p_2}$.
This controller emerges from a standard controller, e.g., LQG or $\mathcal{H}_\infty$-design, where an additional performance input $w_{p_2}$ is introduced.
This provides us later with the opportunity to study the influence of other errors acting on the controller state and input due to quantization and cryptographic noise, which are neglected at first. 
In the following, we use $w_{p} = \begin{pmatrix}w_{p_1}^\top &w_{p_2}^\top\end{pmatrix}^\top$ as the combined performance input.
The interconnection of the system and the controller yields the closed loop depicted in Figure~\ref{fig:blocks}.
As depicted by the padlocks, the measurement signal is encrypted before it is sent to the controller in the cloud, and the control input signal is only decrypted at the plant before it is applied at the actuator.

\begin{figure}[t]
	
	\centering
	\begin{tikzpicture}[scale=1, auto, >=stealth']
		
		\colorlet{effectsColor}{orange}
		
		\def\blockHeight{2\baselineskip}
		\node[block, minimum height=\blockHeight]  at (0,0)  (plant) {\parbox{8em}{\hfill $x$\\}};
		\node[]  at (plant)  (plantText) {\centering Plant};
		\node[block, minimum height=\blockHeight, below = 2.5\baselineskip of plant] (controller) {\parbox{8em}{\hfill $x_c$\\}};
		\node[]  at (controller)  (controllerText) {\centering Controller};

		\node[guide, left=0.5cm of plant] (guideLeft) {};
		\node[guide, right=0.5cm of plant] (guideRight) {};
		\node[guide, left=2cm of plant] (guideLeftLeft) {};
		\node[guide, right=2cm of plant] (guideRightRight) {};
		
		\draw[connector]  ($(guideLeftLeft)+(0,+0.6\baselineskip)$) -- node[pos=0.1]{$w_{p_1}$} ($(plant.west)+(0,+0.6\baselineskip)$);
		\draw[connector]  ($(plant.east)+(0,+0.6\baselineskip)$) -- node[pos=0.9]{$z_p$}  ($(guideRightRight)+(0,+0.6\baselineskip)$);

		\draw[connector]  ($(controller.west)+(0,+0\baselineskip)$) -| node[pos=0.75,left]{$u$} ($(plant.west)+(-0.9cm,-0.6\baselineskip)$) -- node[effectsColor, pos=0.15,anchor=mid,yshift=-0.3\baselineskip] (unlock) {\Huge\faUnlock} ($(plant.west)+(0,-0.6\baselineskip)$);
		\draw[connector]  ($(plant.east)+(0,-0.6\baselineskip)$) -- node[effectsColor, pos=1,anchor=mid,yshift=-0.3\baselineskip] (lock) {\Huge\faLock}  ($(plant.east)+(0.9cm,-0.6\baselineskip)$) |-  node[pos=0.25,right]{$y$} ($(controller.east)+(0,+0.6\baselineskip)$);

		\draw[connector]  ($(controller.east)+(2cm,-0.6\baselineskip)$) 
		-- 
		node[pos=0.1,above]{$w_{p_2}$} ($(controller.east)+(0,-0.6\baselineskip)$);

		\begin{scope}[on behind layer]
			\node[cloud, cloud puffs=17, cloud ignores aspect, align=center, draw, fill = white!80!gray,fit=(controller),inner sep=0.9em] (Controller2){};
		\end{scope}
		
	\end{tikzpicture}
	\caption{Block diagram of the dynamic encrypted control system with performance channels.}\label{fig:blocks}
	
\end{figure}
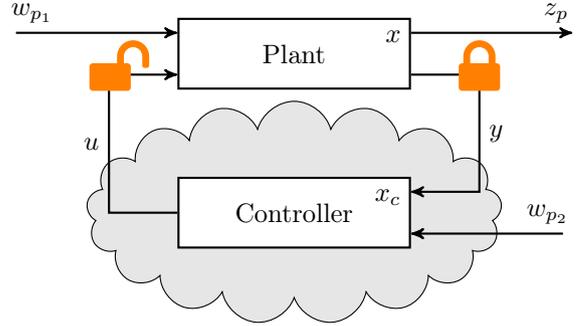

\subsection{Performance goal}
We consider performance as in the following definition.
\begin{definition}
	The system satisfies quadratic performance specified by $P_p$ if it is asymptotically stable for $w_p=0$ and there exists $\epsilon>0$ such that
	\begin{align*}
		\sum_{t=0}^{\infty}
		\begin{pmatrix} 
			w_p(t) \\ z_p(t)
		\end{pmatrix}^\top
		P_p
		\begin{pmatrix} 
			w_p(t) \\ z_p(t)
		\end{pmatrix} &\leq -\epsilon \sum_{t=0}^{\infty}w_p(t)^\top w_p(t)
	\end{align*}
	for $\xi(0)=0$ and all $w_p\in\ell_2$, where $\ell_2$ denotes the space of all square-summable signals. 
\end{definition}
This includes the $\ell_2$-gain, among other common performance specifications.

In the following sections, we describe the existing dynamic encrypted control approaches, followed by a stability and performance analysis based on this definition.
The overall goal is to obtain a stability and performance analysis for all approaches with which the approaches can be compared to the baseline performance.

\section{Dynamic encrypted control approaches}\label{sec:review}

In this section, we review existing approaches from the literature and group them by their methodology in three categories. The first one resolves the problem of a limited number of encrypted multiplications by also adapting the controller to a finite number of multiplications.
In the second category, we describe how the cryptosystem can be modified to enable unlimited multiplications, and in the third category, we group approaches to resolve the problem of the growing scaling factor by integer reformulations that do not need the scaling factor.

\subsection{Limiting the number of multiplications}
In this category, we group three approaches for dynamic encrypted control that modify the dynamic control system such that the state only experiences a finite number of encrypted multiplications before it is re-encrypted or discarded.

\subsubsection{State refreshment by re-encryption}\label{sec:re-enc}
The first proposed approach~\cite{Kogiso2015} for encrypted dynamic control works by not only sending the control input but also the updated state to the plant. The plant decrypts both and uses the control input to regulate the process. The decrypted state is rounded to the required precision and encrypted again. This re-encryption resets the scaling factor and any cryptographic noise in the ciphertext. Together with the sensor measurements, the newly encrypted controller state is sent to the controller at the next time-step, which uses it to replace its previously updated state. As depicted in Figure~\ref{fig:re-encryption}, essentially, the dynamic controller with input $y$ and output $u$ is replaced by a static controller with input $y$ and $x_c$ and output $u$ and the updated controller state, here denoted by $x_c^+$.

While this method is straightforward and easy to implement, it requires the additional communication between controller and plant in both ways, and the additional decryption and encryption by the plant. Of course, this re-encryption could also be done by a trusted third-party, requiring this additional assumption.

The advantages and disadvantages of state refreshment by re-encryption are that\mynobreakpar
\begin{itemize}
	\itemindent=13pt
	\item[\istgreencheck] the exact controller is applied,
	 \item[\istredcross] additional communication is needed.
\end{itemize}

\def\lineShift{0.9cm}
\def\lineShiftDouble{1.5cm}

\begin{figure}[t]
	
	\centering
	\begin{tikzpicture}[scale=1, auto, >=stealth']
		
		\colorlet{effectsColor}{orange}
		
		\def\blockHeight{2\baselineskip}
		\node[block, minimum height=\blockHeight]  at (0,0)  (plant) {\parbox{8em}{\hfill $x$\\}};
		\node[]  at (plant)  (plantText) {\centering Plant};
		\node[block, minimum height=\blockHeight, below = 1.3\baselineskip of plant] (controller) {\parbox{8em}{\hfill $x_c$\\}};
		\node[]  at (controller)  (controllerText) {\centering Controller};

		\node[guide, left=0.5cm of plant] (guideLeft) {};
		\node[guide, right=0.5cm of plant] (guideRight) {};
		\node[guide, left=2cm of plant] (guideLeftLeft) {};
		\node[guide, right=2cm of plant] (guideRightRight) {};

		\draw[connector]  ($(controller.west)+(0,+0.6\baselineskip)$) -| node[pos=0.75,left]{$u$} ($(plant.west)+(-\lineShift,-0.6\baselineskip)$) -- node[effectsColor, pos=0,anchor=mid, yshift=-0.1cm, xshift=0.11cm] (unlocku) {\huge\faUnlock}  ($(plant.west)+(0,-0.6\baselineskip)$);
		\draw[connector]  ($(plant.east)+(0,-0.6\baselineskip)$) -- node[effectsColor, pos=1,anchor=mid, yshift=-0.1cm] (locky) {\huge\faLock} ($(plant.east)+(\lineShift,-0.6\baselineskip)$) |- node[pos=0.25,right]{$y$} ($(controller.east)+(0,+0.6\baselineskip)$);

		\draw[connector]  ($(controller.west)+(0,-0.6\baselineskip)$) -| node[pos=0.75,left]{\color{istred}$x_c^+$} ($(plant.west)+(-\lineShiftDouble,+0.6\baselineskip)$) --node[effectsColor, pos=0,anchor=mid, yshift=-0.1cm, xshift=0.11cm] (unlockx) {\huge\faUnlock} ($(plant.west)+(0,+0.6\baselineskip)$);
		\draw[connector]  ($(plant.east)+(0,+0.6\baselineskip)$) -- node[effectsColor, pos=1,anchor=mid, yshift=-0.1cm] (lockx) {\huge\faLock} ($(plant.east)+(\lineShiftDouble,+0.6\baselineskip)$) |- node[pos=0.25,right]{\color{istblue}$x_c$} ($(controller.east)+(0,-0.6\baselineskip)$);

		\begin{scope}[on behind layer]
			\node[cloud, cloud puffs=17, cloud ignores aspect, align=center, draw, fill = white!80!gray,fit=(controller),inner sep=0.0em] (Controller2){};
		\end{scope}
		
	\end{tikzpicture}
	\caption{Plant supplies cloud with "fresh" encryption of $x_c$}
	\label{fig:re-encryption}
\end{figure}
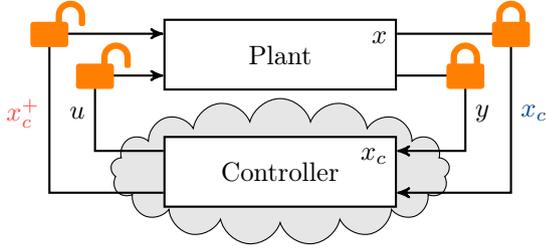

\subsubsection{Resetting controller}\label{sec:reset}
Instead of turning the dynamic controller into a entirely static controller by re-encryption, it was proposed to use the dynamic controller as long as encrypted updates are possible~\cite{Murguia2020}. If the cryptosystem does not allow for further control updates, the controller state is reset to a fresh encryption of zero. Thereby, the controller can continue to run and to provide control inputs; however, with every state-reset, a transient phase is introduced, which reduces control performance. 

The advantages and disadvantages of resetting controllers are that\mynobreakpar
\begin{itemize}
	\itemindent=13pt
	\item[\istgreencheck] no communication is needed,
	\item[\istredcross] repeated transient phases reduce the control performance,
	\item[\istredcross] only strongly stabilizable systems can be controlled.
\end{itemize}

\subsubsection{FIR controller}\label{sec:FIR}
In~\cite{Schluter2021}, a finite-impulse-response (FIR) controller was proposed as an approximation of the dynamic controller that possibly resembles an infinite-impulse-response (IIR) system.
FIR controllers solve the problem of limited encrypted multiplications as their controller state only experiences a finite number of updates or time-shifts, respectively. Every output of the FIR controller consists only of a combination of a finite number of inputs.
In~\cite{Schluter2021}, it was also observed that the previously proposed reset controller can be viewed as a special case of an FIR controller with a periodically growing and resetting horizon. Thus, an FIR controller with constant horizon can lead to better control performance than a resetting controller with the same maximum horizon.
Since FIR controllers are inherently stable, the plant has to be strongly stabilizable, i.e., stabilizable by a stable controller~\cite{Adamek2024a}.
Design methods for FIR controllers were discussed in~\cite{Adamek2024a}.
If an FIR controller leads to an acceptable performance in closed-loop, this is the easiest method to circumvent the problems of dynamic encrypted control, as no specific changes to the implementation have to be made.

The advantages and disadvantages of FIR controllers are that\mynobreakpar
\begin{itemize}
	\itemindent=13pt
	\item[\istgreencheck] no communication is needed,
	\item[\istgreencheck] better performance than with resetting controller can be achieved,
	\item[\istredcross] only approximation of IIR controllers can be applied,
	\item[\istredcross] only strongly stabilizable systems can be controlled.
\end{itemize}

\subsection{Enabling unlimited multiplications: Bootstrapping}\label{sec:Bootstrapping}

The cryptographic solution to extending the limited number of encrypted multiplications is called bootstrapping and was introduced by~\cite{gentry2009fully}. 
It enables modern homomorphic cryptosystems such as CKKS to having an unlimited number of multiplications~\cite{Cheon2018a}.
While the original idea resembled basically the same approach as in Section~\ref{sec:re-enc} but in a purely encrypted fashion, bootstrapping in CKKS relies on a so-called modular reduction, where in theory, the modulo function $\modq$ is evaluated to remove from the ciphertext an unknown number of overflows $rq$ for an $r\in\Z$. 
Since encrypted computations only allow for polynomials to be evaluated and the modulo function requires a piecewise evaluation, in practice, a polynomial approximation of the modulo function in the relevant areas is used in this step. An overview of the different approaches is given in~\cite{Badawi2023,Marcolla2022a}. This necessary approximation, however, introduces errors into the ciphertexts. In Figure~\ref{fig:modPoly}, a part of the required modulo function with its polynomial approximation according to~\cite{schlor24a} is depicted. Figure~\ref{fig:sector} shows the errors between the modulo function and the polynomial, where the different sections of the modulo function are overlayed.
The drawback of this method is the high computational demand, which is why bootstrapping is in the order of seconds, and that errors are introduced into the encrypted values by the approximation of the modulo function by the bootstrapping polynomial.
In~\cite{Kim2016}, timing aspects of bootstrapping were considered and an orchestrated set of redundant controllers was used to bridge the time where one of the controllers performs bootstrapping on its state.
For taming the bootstrapping errors, in~\cite{schlor24a} a robust analysis approach was introduced that guarantees stability and performance despite the possible errors during bootstrapping.

\begin{figure}[t]
	\centering
	\setlength\figurewidth{0.8\columnwidth}
	\setlength\figureheight{0.3\columnwidth}
	\input{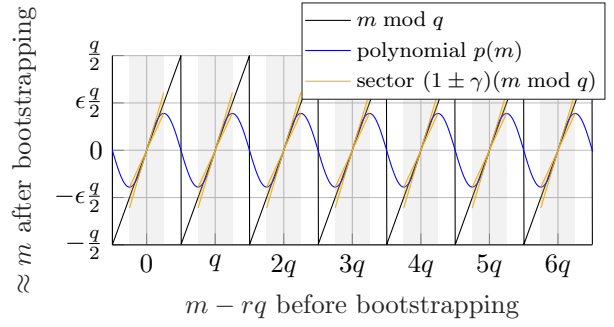} 
	\caption{The modulo function and its polynomial approximation for bootstrapping.}\label{fig:modPoly}
\end{figure}

\begin{figure}[t]
	\centering
	\setlength\figurewidth{0.8\columnwidth}
	\setlength\figureheight{0.3\columnwidth}
	\input{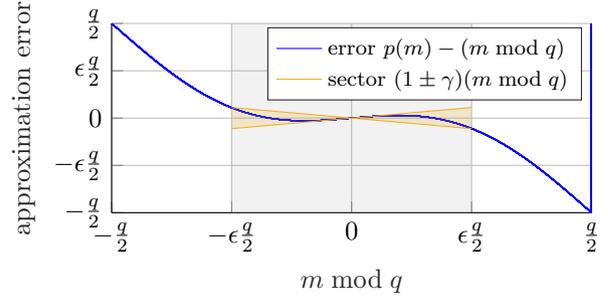} 
	\caption{Relative error of the polynomial approximation to the modulo function for bootstrapping.
		The figure shows multiple error functions since the bootstrapping polynomial in Fig.~\ref{fig:modPoly} is evaluated at different intervals depending on the offset $rq$.
	}\label{fig:sector}
\end{figure}

The advantages and disadvantages of bootstrapping controllers are that\mynobreakpar
\begin{itemize}
	\itemindent=13pt
	\item[\istgreencheck] no communication is needed,
	\item[\istgreencheck] IIR controllers can be applied,
	\item[\istredcross] bootstrapping is computationally demanding.
\end{itemize}

\subsection{Avoiding one of the problems with unlimited multiplications: Integer reformulations}\label{sec:integer}
In this category, we group approaches that reformulate the dynamic controller as a controller with an integer system matrix. With this trick, the encrypted implementation can work without a scaling factor, or the scaling factor is equal to one, respectively. Thus, by repeatedly multiplying a scaling factor of one, no overflows are generated. 
That controllers with integer coefficients are favorable was observed in~\cite{Cheon2018b}, and integer approximations were proposed.
The following approaches implement the nominal controller with an integer implementation using artificial feedback with the plant.

\subsubsection{Integer dynamics by feedback: controller type}
Proposed by~\cite{Kim2021}, the original controller dynamics matrix $A_c$ consisting of fixed-point numbers can be replaced by a system matrix $A_{\Z}$ with integer entries if an additional feedback between the controller and the plant is introduced. The methodology behind this approach is classical pole-placement by state feedback. 
As depicted in Figure~\ref{fig:transcontroller}, the controller sends the artificial output $K x_c$ in addition to the control input $u$ to the plant. The plant decrypts both signals, encrypts $y+K x_c$ again, and sends it to the cloud.
With this procedure, the eigenvalues of the system matrix $A_{\Z}$ can be placed on algebraic integers, which are the roots of monic polynomicals with integer coefficients, by the additional feedback $K x_c$ as
\begin{equation*}
	\begin{aligned}
	x_c(t+1) &= A_cx_c(t) + B_c y(t) + B_c K x_c(t) -B_c K x_c(t)\\
	&= \underbrace{(A_c-B_c K)}_{ A_{\Z} } x_c(t) + B_c(y(t)+K x_c(t))\\
	&= A_{\Z} x_c(t) + B_c(y(t)+K x_c(t)).
\end{aligned}
\end{equation*}
Please note that we have omitted the performance channel here in favor of a more compact representation.
After a state transformation, the matrix $A_{\Z}$ only contains integers and no scaling factor is needed to encode the entries for encryption.
Thus, the possibly unstable or FIR controller with dynamics matrix $A_{\Z}$ is applied and the artificial feedback by the plant leads to the desired closed-loop behavior.

\begin{figure}[t]
	
	\centering
	\begin{tikzpicture}[scale=1, auto, >=stealth']
		
		\colorlet{effectsColor}{orange}
		
		\def\blockHeight{2\baselineskip}
		\node[block, minimum height=\blockHeight]  at (0,0)  (plant) {\parbox{8em}{\hfill $x$\\}};
		\node[]  at (plant)  (plantText) {\centering Plant};
		\node[block, minimum height=\blockHeight, below = 1.3\baselineskip of plant] (controller) {\parbox{8em}{\hfill $x_c$\\}};
		\node[]  at (controller)  (controllerText) {\centering Controller};

		\node[guide, left=0.5cm of plant] (guideLeft) {};
		\node[guide, right=0.5cm of plant] (guideRight) {};
		\node[guide, left=2cm of plant] (guideLeftLeft) {};
		\node[guide, right=2cm of plant] (guideRightRight) {};
		
		\draw[connector]  ($(controller.west)+(0,+0.6\baselineskip)$) -| node[pos=0.75,left]{$u$} ($(plant.west)+(-\lineShift,-0.6\baselineskip)$) -- node[effectsColor, pos=0,anchor=mid, yshift=-0.1cm, xshift=0.11cm] (unlocku) {\huge\faUnlock}  ($(plant.west)+(0,-0.6\baselineskip)$);
		\draw[connector]  ($(plant.east)+(0,-0\baselineskip)$) -- node[effectsColor, pos=1,anchor=mid, yshift=-0.1cm] (locky) {\huge\faLock} ($(plant.east)+(\lineShift,-0\baselineskip)$) |- node[pos=0.25,right]{$y+K x_c$} ($(controller.east)+(0,+0\baselineskip)$);

		\draw[connector]  ($(controller.west)+(0,-0.6\baselineskip)$) -| node[pos=0.75,left]{$Kx_c$} ($(plant.west)+(-\lineShiftDouble,+0.6\baselineskip)$) --node[effectsColor, pos=0,anchor=mid, yshift=-0.1cm, xshift=0.11cm] (unlockx) {\huge\faUnlock} ($(plant.west)+(0,+0.6\baselineskip)$);
		\begin{scope}[on behind layer]
			\node[cloud, cloud puffs=17, cloud ignores aspect, align=center, draw, fill = white!80!gray,fit=(controller),inner sep=0.0em] (Controller2){};
		\end{scope}
		
	\end{tikzpicture}
\caption{Pole shift by $K x_c$ (and transformation) such that $A_{\Z}$ is an integer matrix.}
\label{fig:transcontroller}
\end{figure}

\subsubsection{Integer dynamics by feedback: observer type}
We observe that a similar effect can be obtained if pole-placement via an observer-type feedback is introduced.
In this case, the controller only sends the control input $u$ to the plant, but the plant computes $Lu$, and sends it back to the controller with the measurement $y$ as depicted in Figure~\ref{fig:transobserver}.
By this additional feedback, similarly to the previous approach, the eigenvalues of $A_{\Z}$ can be placed on algebraic integers with $u(t)=C_c x_c(t)$ as
\begin{equation*}
	\begin{aligned}
	x_c(t+1) &= A_cx_c(t) + B_c y(t) - L C_c x_c(t) +Lu(t),\\
	&= \underbrace{(A_c-L C_c)}_{ A_{\Z} } x_c(t) + B_cy(t)+Lu(t)\\
	&= A_{\Z} x_c(t) + B_cy(t)+Lu(t).
\end{aligned}
\end{equation*}
Depending on the input and output dimensions, the controller or observer type integer reformulation can be favorable.

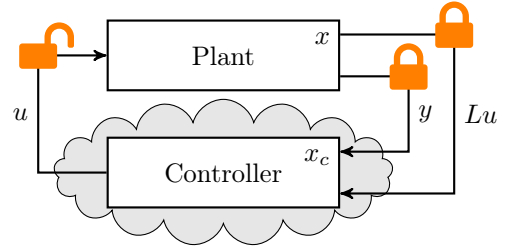
\begin{figure}[t]
	
	\centering
	\begin{tikzpicture}[scale=1, auto, >=stealth']
		
		\colorlet{effectsColor}{orange}
		
		\def\blockHeight{2\baselineskip}
		\node[block, minimum height=\blockHeight]  at (0,0)  (plant) {\parbox{8em}{\hfill $x$\\}};
		\node[]  at (plant)  (plantText) {\centering Plant};
		\node[block, minimum height=\blockHeight, below = 1.3\baselineskip of plant] (controller) {\parbox{8em}{\hfill $x_c$\\}};
		\node[]  at (controller)  (controllerText) {\centering Controller};

		\node[guide, left=0.5cm of plant] (guideLeft) {};
		\node[guide, right=0.5cm of plant] (guideRight) {};
		\node[guide, left=2cm of plant] (guideLeftLeft) {};
		\node[guide, right=2cm of plant] (guideRightRight) {};
		
		\draw[connector]  ($(controller.west)+(0,+0\baselineskip)$) -| node[pos=0.75,left]{$u$} ($(plant.west)+(-\lineShift,-0\baselineskip)$) -- node[effectsColor, pos=0,anchor=mid, yshift=-0.1cm, xshift=0.11cm] (unlocku) {\huge\faUnlock}  ($(plant.west)+(0,-0\baselineskip)$);
		\draw[connector]  ($(plant.east)+(0,-0.6\baselineskip)$) -- node[effectsColor, pos=1,anchor=mid, yshift=-0.1cm] (locky) {\huge\faLock} ($(plant.east)+(\lineShift,-0.6\baselineskip)$) |- node[pos=0.25,right]{$y$} ($(controller.east)+(0,+0.6\baselineskip)$);

		\draw[connector]  ($(plant.east)+(0,+0.6\baselineskip)$) -- node[effectsColor, pos=1,anchor=mid, yshift=-0.1cm] (lockx) {\huge\faLock} ($(plant.east)+(\lineShiftDouble,+0.6\baselineskip)$) |- node[pos=0.25,right]{$Lu$} ($(controller.east)+(0,-0.6\baselineskip)$);

		\begin{scope}[on behind layer]
			\node[cloud, cloud puffs=17, cloud ignores aspect, align=center, draw, fill = white!80!gray,fit=(controller),inner sep=0.0em] (Controller2){};
		\end{scope}
		
	\end{tikzpicture}
\caption{Pole shift by $Lu$ (and transformation) such that $A_{\Z}$ is integer matrix.}
\label{fig:transobserver}
\end{figure}

The advantages and disadvantages of integer controllers by feedback are that\mynobreakpar
\begin{itemize}
	\itemindent=13pt
	\item[\istgreencheck] the exact controller is applied,
	\item[\istgreencheck] less communication than by state re-encryption is needed,
	\item[\istredcross] additional communication is needed,
		\item[\istredcross] the wanted controller is only achieved by feedback with the plant.
\end{itemize}

\subsubsection{Integer conversion by scaling}

Different from the approaches for integer reformulation using feedback, in~\cite{Lee2023b,Lee2025}, a non-minimal realization of controllers is found that has integer coefficients. The iterative design method for these type of controllers does not require an artificial feedback with the plant.
However, it introduces stable pole-zero cancellations and typically requires a much larger dimension of the converted controller matrices.

{The advantages and disadvantages of integer controllers by scaling are that\mynobreakpar
\begin{itemize}
	\itemindent=13pt
	\item[\istgreencheck] the exact controller is applied,
	\item[\istgreencheck] no additional communication is needed,
	\item[\istredcross] a non-minimal realization is applied.
\end{itemize}}

\section{Stability and performance analysis} \label{sec:analysis}

After listing the available approaches to encrypted dynamic control, we analyze the performance of the different control strategies.

\subsection{Baseline for nominal system and exact methods}
First, we show how the baseline performance can be determined that is achieved by the nominal controller. Neglecting quantization errors, the exact encrypted dynamic control approaches in Sections~\ref{sec:re-enc} and~\ref{sec:integer} also achieve this performance.

The interconnection of the plant~\eqref{eq:plant} and the controller~\eqref{eq:controller} with joint state $\xi = \begin{pmatrix}
	x \\ x_c
\end{pmatrix}$ 
results in the closed-loop system representation
\begin{align}\label{eq:clsys}
	\left(\begin{array}{c} 
		\xi(t+1) \\ \hline  z_p(t)
	\end{array}\right)
	&=
	\left(\begin{array}{c | c } 
		\Abc 	& \Bbc_p \\ 
		\hline  
		\Cbc_p & \Dbc_{pp}
	\end{array}\right)
	\left(\begin{array}{c} 
		\xi(t)\\ \hline w_{p}(t)
	\end{array}\right)
\end{align}
with
\begin{multline*}
	\left(\begin{array}{c | c } 
		\Abc 	& \Bbc_p \\ 
		\hline  
		\Cbc_p & \Dbc_{pp} 
	\end{array}\right)
	=\\
	\left(\begin{array}{cc | c c } 
		A + B D_c C  & B C_c 	& B_1 + BD_cF_1 & BF_2 		 \\ 
		B_c C & A_c 				&  B_cF_1 & B_2 					 \\ 
		\hline  
		C_1 + ED_cC & EC_c 		& D_1 + ED_cF_1 & EF_2 		
	\end{array}\right).
\end{multline*}

Then, by applying well-known results from control theory~\cite[Prop.\ 3.9]{scherer2000linear}, we can obtain the following theorem.
\begin{theorem}{\cite[Prop.\ 3.9]{scherer2000linear}}\label{thm:1}
	The encrypted closed-loop system~\eqref{eq:clsys} satisfies quadratic performance with performance index $P_p = 	
	\begin{pmatrix} 
		Q_p & S_p\\
		S_p^\top & R_p
	\end{pmatrix} $ with $R_p \succeq 0$, if 
	there exist $X \succ 0$ such that
	\begin{align*}
		(\star)^\top  \!
		\begin{pmatrix} 
			-X & 0\\
			0 & X
		\end{pmatrix} 
		\begin{pmatrix} 
			I & 0 \\
			\Abc & \Bbc_p 
		\end{pmatrix} 
		+
		(\star)^\top  \!
		P_p
		\begin{pmatrix} 
			0 & I \\
			\Cbc_p & \Dbc_{pp} 
		\end{pmatrix}  &{}
	\!
		\prec 0.
	\end{align*}
	\vspace*{-0.8\baselineskip}
\end{theorem}

For our analysis of the specific encrypted dynamic control approaches in the next section, we need one more preliminary result.
It concerns an equivalent lifted system description, where the original system is only sampled at time instants $t=kT$ with $T\in\N$ and a new time-index $k\in\N$.

With the lifted state and performance signals
\begin{equation*}
	\begin{aligned}
	\tilde{z}_p(k) &=
	\begin{pmatrix}
		z_p(kT)\\
		\vdots\\
		z_p(kT+(T-1))
	\end{pmatrix}, &
\tilde{\xi}(k) &= \xi(kT),
\\
	\tilde{w}_p(k) &=
	\begin{pmatrix}
		w_p(kT)\\
		\vdots\\
		w_p(kT+(T-1))
	\end{pmatrix},&&
\end{aligned}
\end{equation*}
\noindent the lifted system with time-index $k$ can be represented by
\begin{align}\label{eq:liftSys}
	\left(\begin{array}{c} 
		\tilde{\xi}(k+1) \\ \hline  \tilde{z}_p(k) 
	\end{array}\right)
	&=
	\left(\begin{array}{c | c } 
		\AbcTilde
		& \tilde{\Bbc}_p  \\ 
		\hline  
		\tilde{\Cbc}_p & \tilde{\Dbc}_{pp}
	\end{array}\right)
	\left(\begin{array}{c} 
		\tilde{\xi}(k)\\ \hline \tilde{w}_{p}(k) 
	\end{array}\right)
\end{align}
with
\begin{multline*}
	\left(\begin{array}{c | c } 
		\AbcTilde
		& \tilde{\Bbc}_p  \\ 
		\hline  
		\tilde{\Cbc}_p & \tilde{\Dbc}_{pp}
	\end{array}\right)
	=\\
	\left(\begin{array}{c | cccc  } 
		\Abc^{T} 
		& 
			\Abc^{T-1}\Bbc_p &\dots&\dots & \Bbc_p
		\\ 
		\hline
			\Cbc_p 		&\Dbc_{pp} & 0 & \cdots & 0 \\
			\vdots		& \Cbc_p \Bbc_{p} & \multicolumn{2}{c}{\smash{\raisebox{-0.6\normalbaselineskip}{\diagdots[-25]{3.7\normalbaselineskip}{.5em}\hspace*{0.7em}}}} & \vdots \\
			\vdots		& \vdots & \multicolumn{1}{c}{\smash{\raisebox{-0.2\normalbaselineskip}{\diagdots[-25]{2.5\normalbaselineskip}{.5em}\hspace*{0em}}}} & \multicolumn{1}{c}{\smash{\raisebox{0.9\normalbaselineskip}{\hspace*{1.3em}\diagdots[-25]{2.5\normalbaselineskip}{.5em}}}} & 0 \\
			\Cbc_p \Abc^{T-1}		& \Cbc_p \Abc^{T-2} \Bbc_{p} & \mathclap{\cdots} & \Cbc_p \Bbc_{p} & \Dbc_{pp}
	\end{array}\right).
\end{multline*}

\begin{lemma}{\cite[Lem.\ 3]{Lang2024}}\label{lem:eq}
	The original system~\eqref{eq:clsys} satisfies quadratic performance specified by
	\vspace*{-0.2\baselineskip}
	\begin{align*}
		P_p = 	
		\begin{pmatrix} 
			Q_p & S_p\\
			S_p^\top & R_p
		\end{pmatrix},
	\end{align*}
	if and only if the lifted system~\eqref{eq:liftSys} satisfies quadratic performance specified by \begin{align*}
		\tilde{P}_p = 
		\begin{pmatrix} 
			I_{T}\otimes Q_p & I_{T}\otimes S_p\\
			I_{T}\otimes S_p^\top & I_{T}\otimes R_p
		\end{pmatrix} .
	\end{align*}
\end{lemma}

With this result, we can analyze the system if errors are introduced every $T$ time steps due to bootstrapping approximations or controller state resets.

\subsection{Bootstrapping}

As shown in~\cite{schlor24a}, the bootstrapping error that results from the deviation between the modulo function and the bootstrapping polynomial $p$ can be treated as a static, time-varying uncertainty $\Delta_r$, that depends on the unknown number of overflows $r$, acting on the entire controller state during bootstrapping. 
For the analysis, an uncertainty channel from $z_u= x_c$ to $w_u = \Delta_r(z_u) = p(x_c)-(x_c \modq)$ can be introduced, where the polynomial $p$ is applied component-wise.

 Due to the particular bootstrapping polynomial in~\cite{schlor24a} with relative error bounds, every possible bootstrapping uncertainty $\Delta_r$ can be treated as an element of the set $\boldsymbol{\Delta}$ of all uncertainties satisfying the same relative error bound.
Every sector bounded uncertainty $\Delta$ in the uncertainty set $\boldsymbol{\Delta}$ satisfies 
 \begin{align}\label{eq:sectorP}
 	\begin{pmatrix} 
 		\Delta(z_u) \\ z_u
 	\end{pmatrix}^\top
 	\tau P_u
 	\begin{pmatrix} 
 		\Delta(z_u) \\ z_u
 	\end{pmatrix} &\geq 0
 \end{align}
 with any $\tau>0$ and the multiplier
 \begin{align*}%
 	P_u = 
 	\begin{pmatrix} 
 		-2I & 0 \\
 		0 & 2 \gamma^2 I
 	\end{pmatrix}.
 \end{align*}
The sector bound together with the actual bootstrapping errors is depicted in Figure~\ref{fig:sector}.

This uncertainty description together with Lemma~\ref{lem:eq} yields a performance theorem for dynamic encrypted control with bootstrapping based on robust control theory. 
 \begin{theorem}{\cite[Thm.\ 2]{schlor24a}}\label{thm:2}
 	The encrypted closed-loop system~\eqref{eq:clsys} with the bootstrapping uncertainty~\eqref{eq:sectorP} satisfies robust quadratic performance with performance index $P_p = 	
 	\begin{pmatrix} 
 		Q_p & S_p\\
 		S_p^\top & R_p
 	\end{pmatrix} $ with $R_p \succeq 0$, if 
 	there exist $X \succ 0$ and $\tau>0$  such that
 	\begin{align*}
 		(\star)^\top  
 		\begin{pmatrix} 
 			-X & 0\\
 			0 & X
 		\end{pmatrix} 
 		\begin{pmatrix} 
 			I & 0 & 0\\
 			\AbcTilde & \tilde{\Bbc}_p & \tilde{\Bbc}_u
 		\end{pmatrix}  &{}\\
 		+
 		(\star)^\top  
 		\tilde{P}_p
 		\begin{pmatrix} 
 			0 & I & 0\\
 			\tilde{\Cbc}_p & \tilde{\Dbc}_{pp} & \tilde{\Dbc}_{pu}
 		\end{pmatrix}  &{}\\
 		+
 		(\star)^\top  
 		\tau P_u
 		\begin{pmatrix} 
 			0 & 0 & I\\
 			\Cbc_u & 0 & 0
 		\end{pmatrix}  &{}
 		\prec 0
 	\end{align*}
 	\noindent with
 	\begin{equation*}
 		\tilde{P}_p = 
 		\begin{pmatrix} 
 			I_{T}\otimes Q_p & I_{T}\otimes S_p\\
 			I_{T}\otimes S_p^\top & I_{T}\otimes R_p
 		\end{pmatrix} %
 	\end{equation*}
 and
 \begin{align*}
 	\tilde{\Bbc}_u &= \Abc^{T-1}\Bbc_u,
 	&
 	 \Cbc_u &= \begin{pmatrix}
 		0 & I
 	\end{pmatrix}, 
 \\
 \Bbc_u &= \begin{pmatrix}
 	0\\
 	A_c
 \end{pmatrix}, &
 	\tilde{\Dbc}_{pu} &= \begin{pmatrix}
 		0\\
 		\Cbc_p \Bbc_{u} \\
 		\vdots\\
 		\Cbc_p \Abc^{T-2} \Bbc_u
 	\end{pmatrix}.
 \end{align*}
 	\vspace*{-0.8\baselineskip}
 \end{theorem}

With this result, performance and stability can be certified for the encrypted control system including bootstrapping.

\subsection{Resetting controllers}

In the original paper~\cite{Murguia2020}, an asymptotic stability analysis was provided. Here, we extend it to quadratic performance.
Resetting the controller state to zero instead of bootstrapping similarly introduces an error in the controller state. Whereas the bootstrapping error is given by the specific deviation of the bootstrapping polynomial from the modulo function, the reset introduces an error of exactly minus the identity.
This error can be captured in the same analysis framework with the multiplier
\begin{align}\label{eq:PReset}
	P_u = 
	\begin{pmatrix} 
		-2I & -2I \\
		-2I & -2I
	\end{pmatrix}.
\end{align}
Thus, only $\Delta(z_u) = -z_u$ satisfies the uncertainty description~\eqref{eq:sectorP}.

Based on Theorem~\ref{thm:2} we can state a performance theorem for dynamic encrypted control with periodic controller state resets. 
\begin{theorem}\label{thm:reset2}
	The encrypted closed-loop system~\eqref{eq:clsys} with a controller state reset every $T$ time-steps satisfies quadratic performance with performance index $P_p = 	
	\begin{pmatrix} 
		Q_p & S_p\\
		S_p^\top & R_p
	\end{pmatrix} $ with $R_p \succeq 0$, 
if the conditions of Theorem~\ref{thm:2} hold with the multiplier $P_u$ from~\eqref{eq:PReset}.
\end{theorem}
 \begin{proof}
	 	The proof is a direct consequence of Theorem~\ref{thm:2} with the alternative uncertainty description.
	 \end{proof}
 
By observing that in this case, the reset error is actually known, a nominal performance test can also be used.
For this, we construct the lifted system in such a way that over the horizon of $T$ time-steps the whole system behavior including the reset is captured.
The lifted resetting system with time-index $k$ can then be represented by
\begin{align}\label{eq:liftSysReset}
	\left(\begin{array}{c} 
		\tilde{\xi}(k+1) \\ \hline  \tilde{z}_p(k) 
	\end{array}\right)
	&=
	\left(\begin{array}{c | c } 
		\AbcTilde_r
		& \tilde{\Bbc}_{p}  \\ 
		\hline  
		\tilde{\Cbc}_{p} & \tilde{\Dbc}_{pp}
	\end{array}\right)
	\left(\begin{array}{c} 
		\tilde{\xi}(k)\\ \hline \tilde{w}_{p}(k) 
	\end{array}\right)
\end{align}
with $\tilde{\Bbc}_{p}$, $\tilde{\Cbc}_{p}$, and $\tilde{\Dbc}_{pp}$ as in~\eqref{eq:liftSys}, and
\bgroup \allowdisplaybreaks
\begin{align*}
	\AbcTilde_r &= \Abc^{T-1}\Abc_r, 
	&
	\Abc_r &=  \begin{pmatrix}
		A + B D_c C  & B C_c 	 \\ 
		B_c C & 0 				 
	\end{pmatrix}.
\end{align*}
\egroup

Then, the nominal performance test for the control system with the resetting controller is as follows.
\begin{theorem}\label{thm:reset1}
	The encrypted closed-loop system~\eqref{eq:clsys} with a controller state reset every $T$ time-steps satisfies quadratic performance with performance index $P_p = 	
	\begin{pmatrix} 
		Q_p & S_p\\
		S_p^\top & R_p
	\end{pmatrix} $ with $R_p \succeq 0$, 
	if and only if the lifted system~\eqref{eq:liftSysReset} satisfies quadratic performance specified by \begin{align*}
	\tilde{P}_p = 
	\begin{pmatrix} 
		I_{T}\otimes Q_p & I_{T}\otimes S_p\\
		I_{T}\otimes S_p^\top & I_{T}\otimes R_p
	\end{pmatrix} .
\end{align*}

\end{theorem}
\begin{proof}
	The proof follows from Lemma~\ref{lem:eq} and Theorem~\ref{thm:1}.
\end{proof}

\subsection{FIR approximations}

There are several approximation methods to obtain an FIR controller from a pre-designed IIR controller. For an overview, see~\cite{Schluter2021}. There are also multiple direct design approaches for FIR controllers, e.g.,~\cite{Adamek2024a}. 
Here, we focus on FIR approximations of the given IIR controller by a window-based truncation method, where the impulse response of the IIR and FIR controller are matched up to a time horizon that is given by the FIR filter order. 
If such an FIR controller is given, a straightforward nominal performance analysis as in Theorem~\ref{thm:1} can be conducted.
Alternatively, an interpretation similar to our approach for bootstrapping or resetting controllers can be taken. For bootstrapping or resetting controllers the actually implemented controller was interpreted as the nominal IIR controller with a particular uncertainty that accounts for the bootstrapping error or the reset.
Here, for an FIR approximation, the same method can be applied. The FIR approximation can be written as the nominal IIR controller, where in every time-step $t$ the effect of the measurement from time-step $t-T$ on the controller state is removed.
The FIR controller dynamics can thus be written as
\begin{align*}
	x_c(t+1) &= A_{c,\text{FIR}} x_c(t) + B_c y(t) + B_2w_{p_2}(t)\\
	&= A_c x_c(t) + B_c y(t) + B_2w_{p_2}(t)\notag\\
	&\phantom{=~} - A_c^N B_c y(t-T)\\
	u(t) &= C_c x_c(t) + D_c y(t) + F_2 w_{p_2}(t),
\end{align*}
or, alternatively, with the extended controller state ${x}_{y,c}(t) = \begin{pmatrix}
	y(t-T)\\
	\vdots\\
	y(t-1)\\
	x_c(t)
\end{pmatrix}$ and the upper shift matrix $U_{n_y}$ with ones on the $n_y$-th superdiagonal 
\begin{equation}
	\begin{aligned}
	{x}_{y,c}(t+1) =&
	\begin{pmatrix}
		U_{n_y} & 0\\
		\begin{pmatrix}
			- A_c^T B_c & 0
		\end{pmatrix} & A_c
	\end{pmatrix}
	{x}_{y,c}(t)\\
	{}&+ \begin{pmatrix}
		0\\
		I\\
		B_c
	\end{pmatrix} y(t) + \begin{pmatrix}
		0\\
		B_2
	\end{pmatrix}w_{p_2}(t)
	\\
	u(t) =&\begin{pmatrix}
		 0 &C_c
	\end{pmatrix} {x}_{y,c}(t) + D_c y(t) + F_2 w_{p_2}(t).
\end{aligned}\label{eq:FIR}
\end{equation}
We note that this is not a minimal realization of the FIR controller but the construction does not require the explicit evaluation of the impulse response and makes the FIR approximation interpretable with respect to the original given IIR controller.
The FIR controller in~\eqref{eq:FIR} is in a form that is also suitable for our stability and performance tests.

\section{Numerical comparison}~\label{sec:numerics}
For the numerical evaluation and comparison of the approaches we take a benchmark system, a model of a chemical batch reactor, introduced by~\cite{Green1995}, which was already used in the related literature~\cite{Murguia2020,Schluter2021,schlor24a, Adamek2024a}.
The discretized system with sampling period $h=0.1$ and an added performance channel is given by

\vspace*{-0.8\baselineskip}
			\begin{align*}
				A &= \begin{pmatrix}
					1.18 & 0 & 0.51 & -0.4\\ -0.05 & 0.66 & -0.01 & 0.06\\ 0.08 & 0.34 & 0.56 & 0.38\\ 0 & 0.34 & 0.09 & 0.85
				\end{pmatrix},\\
				B &= \begin{pmatrix}
					0 & 0.088\\ -0.47 & -0.001\\ -0.21 & 0.24\\ -0.21 & 0.016
				\end{pmatrix},\\
				C &= C_1 = \begin{pmatrix}
					1 & 0 & 1 & -1\\ 0 & 1 & 0 & 0
				\end{pmatrix},\\
				B_1 &= I, ~~
				F_1 = 0, ~~
				E = 0, ~~
				D_1 = 0
				.
			\end{align*}

		\noindent An $\mathcal{H}_\infty$ controller synthesis for the system yields the controller matrices 
		\begin{align*}
			A_c &=\begin{pmatrix} 0.33 & -0.034 & -0.26 & 0.33\\ 0 & 0 & 0 & 0\\ -0.37 & -0.17 & 0.31 & 0.52\\ -0.035 & -0.2 & 0.051 & 0.86 \end{pmatrix},\\
			B_c &= \begin{pmatrix}
				0.49 & 0.03\\ 0 & 0\\ -0.5 & 0.21\\ -0.011 & 0.24 
			\end{pmatrix},\\
			C_c &= \begin{pmatrix}
				-0.045 & -0.022 & 0.039 & 0.07\\ -0.5 & 0.11 & 0.39 & -0.75
			\end{pmatrix},\\
			D_c &= \begin{pmatrix}
				-0.054 & 1.4\\ -3.6 & -0.09
			\end{pmatrix},
		\end{align*}
	\noindent where we omitted the performance input to the controller, i.e., $B_2$ and $F_2$ do not exist.

As a performance measure, we used the $\ell_2$-gain with $Q_p = -\gamma_{\ell_2}^2I,  S_p = 0$, and $R_p = I$. 
To compare the different approaches, we analyzed the nominal closed-loop with Theorem~\ref{thm:1}, the system with bootstrapping with Theorem~\ref{thm:2}, the system with resets using the robust Theorem~\ref{thm:reset2} and the switched nominal Theorem~\ref{thm:reset1}, and the FIR approximation with Theorem~\ref{thm:1}.
For that, we used a common bootstrapping and resetting time, and FIR horizon of $T=10$, and used the sector bound of $\gamma = 0.223$ for the bootstrapping analysis.
The solutions to the LMIs were found using the MOSEK solver~\cite{mosek}, YALMIP~\cite{Lofberg2004} and MATLAB.
In addition, simulations of the respective approaches were conducted over 10.000 time-steps using the same random disturbance as performance input signal. This provides a lower bound on the actual $\ell_2$-gain. The results are given in Table~\ref{tab:results}.

\begin{table}[tb]
	\centering
\begin{tabularx}{0.7\columnwidth}{l|l}
	Approach & Performance\\
	\hline
	Nominal (Thm.\ \ref{thm:1}) &  1.8030\\
	Bootstrapping (Thm.\ \ref{thm:2}) & 2.0332 \\
	Reset (Thm.\ \ref{thm:reset2}) &  1.9508\\
	Reset (Thm.\ \ref{thm:reset1}) & 1.8869 \\
	FIR (Thm.\ \ref{thm:1}) & 1.8819 \\
	Nominal (simulation) &  1.0301\\
	Bootstrapping (simulation) &  1.0302\\
	Reset (simulation) &  1.0311\\
	FIR (simulation) & 1.0303 \\
\end{tabularx}
\caption{Comparison of the $\ell_2$-gains obtained for the approaches.}
\label{tab:results}
\end{table}

The results demonstrate that with these approaches, performance guarantees can be obtained using the theorems in Section~\ref{sec:analysis}. 
For this example, the FIR approximation lead to a better performance guarantee than the resetting controllers, and the nominal performance analysis theorems yield less conservative results than the robust analysis theorems.

\section{Summary and outlook}~\label{sec:Summary}

In this paper, we have presented an overview of existing methods for dynamic encrypted control, structured by their underlying solution approach.
We have shown how bootstrapping, periodically resetting controllers, and FIR controllers can be analyzed using a robust and a nominal performance test, and how the introduced errors can be compared between the different approaches.
Our numerical benchmark highlights the differences in the performance between the approaches, which is, next to the simplicity of the encrypted implementation, relevant for the selection of the algorithm.

Further research in cryptography could focus on more efficient and precise bootstrapping methods, while research in control could analyze further errors due to encryption and quantization, and if other bootstrapping polynomials are used.

\begin{funding}
	Funded by the Deutsche Forschungsgemeinschaft (DFG, German Research Foundation) under Germany's Excellence Strategy -- EXC 2075 -- 390740016 and within grant AL 316/13-2 -- 285825138 and AL 316/15-1 -- 468094890. 
	S.\ Schlor thanks the Graduate Academy of the SC SimTech for its support.
\end{funding}

\bibliographystyle{unsrt}

\bibliography{BibForAT} %

\vskip10pt
\begin{wrapfigure}{l}{25mm} 
	\includegraphics[width=1in,height=1.25in,clip,keepaspectratio]{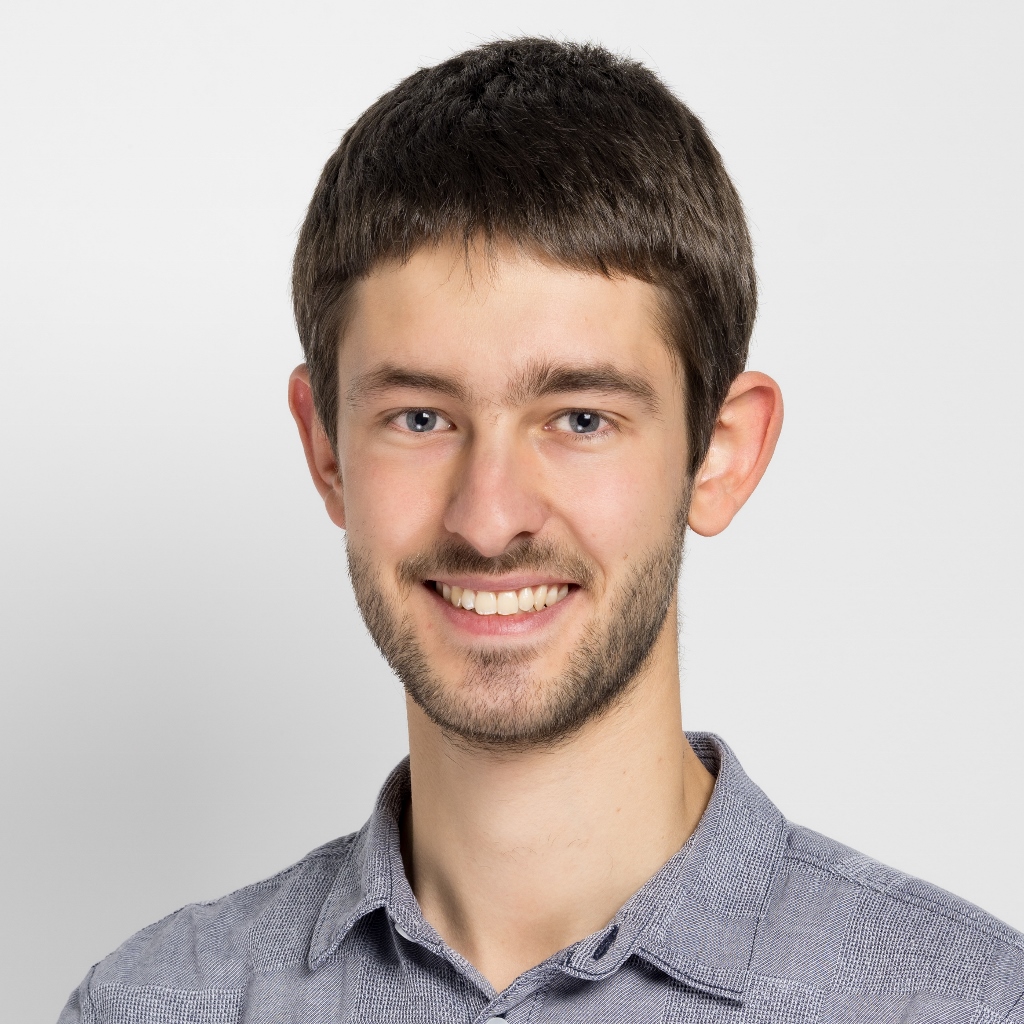}
\end{wrapfigure}\par
\textbf{Sebastian Schlor} received his B.\ Sc.\ and M.\ Sc.\ degree in Engineering Cybernetics from the University of Stuttgart, Stuttgart, Germany, in 2017 and 2020, respectively. Since 2020, he has been a Ph.D.\ student at the University of Stuttgart with the Institute for Systems Theory and Automatic Control under the supervision of Prof.\ Frank Allgöwer.
His research interests include the area of privacy and security of dynamical systems and encrypted control and optimization.
Sebastian Schlor received the Best Student Paper Prize 2025 of the IEEE CSS TC Security and Privacy.\par
\vskip10pt

\begin{wrapfigure}{l}{25mm} 
	\includegraphics[width=1in,height=1.25in,clip,keepaspectratio]{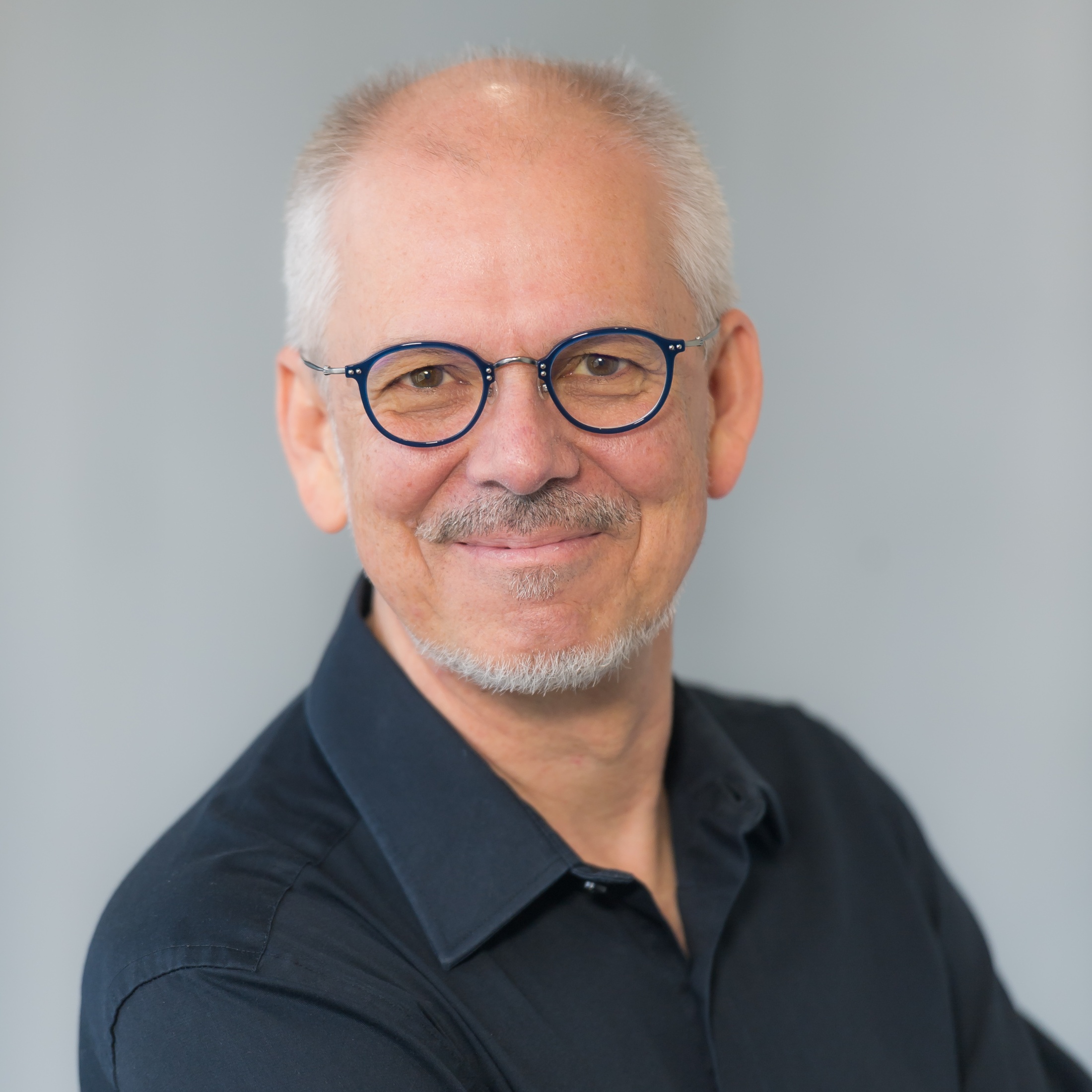}
\end{wrapfigure}\par
\textbf{Frank Allg\"ower} studied engineering cybernetics and applied mathematics in Stuttgart and with the University of California, Los Angeles (UCLA), CA, USA, respectively, and received the Ph.D. degree from the University of Stuttgart, Stuttgart, Germany. Since 1999, he has been the Director of the Institute for Systems Theory and Automatic Control and a professor with the University of Stuttgart. His research interests include predictive control, data-based control, networked control, cooperative control, and nonlinear control with application to a wide range of fields including systems biology. Dr. Allgöwer was the President of the International Federation of Automatic Control (IFAC) in 2017–2020 and the Vice President of the German Research Foundation DFG in 2012–2020.\par
\vskip10pt

\end{document}

%% file: mypackages.tex
\usepackage{xcolor}
\usepackage{fontawesome}
\usepackage{mathtools}
\usepackage{bm}
\usepackage{booktabs}
\usepackage{microtype} %
\usepackage{siunitx}
\usepackage{tabularx, tabulary}

\usepackage{tikz}
\newcommand\tikznode[3][]%
{\tikz[remember picture,baseline=(#2.base)]
	\node[minimum size=0pt,inner sep=0pt,#1](#2){#3};%
}
\usepackage{calc}
\usetikzlibrary{calc} 

\usetikzlibrary{decorations} %
\usetikzlibrary{matrix} %
\usetikzlibrary{arrows} %
\usetikzlibrary{backgrounds, fit}
\usepackage{tikz-layers}

\usetikzlibrary{positioning,shapes.multipart,shapes.arrows, decorations.markings,arrows.meta,shapes.symbols}

\tikzset{
	-|-/.style={
		to path={
			(\tikztostart) -| ($(\tikztostart)!#1!(\tikztotarget)$) |- (\tikztotarget)
			\tikztonodes
		}
	},
	-|-/.default=0.5,
	|-|/.style={
		to path={
			(\tikztostart) |- ($(\tikztostart)!#1!(\tikztotarget)$) -| (\tikztotarget)
			\tikztonodes
		}
	},
	|-|/.default=0.5,
}

\tikzstyle{block} = [draw=black,fill=white,rectangle,thick,minimum height=2em,minimum width=4em, align=center]
\tikzstyle{sum} = [draw,circle,inner sep=0mm,minimum size=2mm]
\tikzstyle{connector} = [->,thick]
\tikzstyle{dashedconnector} = [->,thick,dashed]
\tikzstyle{line} = [thick]
\tikzstyle{dashedline} = [thick,dashed]
\tikzstyle{branch} = [circle,inner sep=0pt,minimum size=1mm,fill=black,draw=black]
\tikzstyle{guide} = []
\tikzstyle{snakeline} = [connector, decorate, decoration={pre length=0.2cm,
	post length=0.2cm, snake, amplitude=.4mm,
	segment length=2mm},thick, magenta, ->]

\newcommand{\gettikzxy}[3]{%
	\tikz@scan@one@point\pgfutil@firstofone#1\relax
	\edef#2{\the\pgf@x}%
	\edef#3{\the\pgf@y}%
}

\makeatletter
\pgfkeys{/pgf/.cd,
	parallelepiped offset x/.initial=2mm,
	parallelepiped offset y/.initial=2mm
}
\pgfdeclareshape{parallelepiped}
{
	\inheritsavedanchors[from=rectangle] %
	\inheritanchorborder[from=rectangle]
	\inheritanchor[from=rectangle]{north}
	\inheritanchor[from=rectangle]{north west}
	\inheritanchor[from=rectangle]{north east}
	\inheritanchor[from=rectangle]{center}
	\inheritanchor[from=rectangle]{west}
	\inheritanchor[from=rectangle]{east}
	\inheritanchor[from=rectangle]{mid}
	\inheritanchor[from=rectangle]{mid west}
	\inheritanchor[from=rectangle]{mid east}
	\inheritanchor[from=rectangle]{base}
	\inheritanchor[from=rectangle]{base west}
	\inheritanchor[from=rectangle]{base east}
	\inheritanchor[from=rectangle]{south}
	\inheritanchor[from=rectangle]{south west}
	\inheritanchor[from=rectangle]{south east}
	\backgroundpath{
		\southwest \pgf@xa=\pgf@x \pgf@ya=\pgf@y
		\northeast \pgf@xb=\pgf@x \pgf@yb=\pgf@y
		\pgfmathsetlength\pgfutil@tempdima{\pgfkeysvalueof{/pgf/parallelepiped
				offset x}}
		\pgfmathsetlength\pgfutil@tempdimb{\pgfkeysvalueof{/pgf/parallelepiped
				offset y}}
		\def\ppd@offset{\pgfpoint{\pgfutil@tempdima}{\pgfutil@tempdimb}}
		\pgfpathmoveto{\pgfqpoint{\pgf@xa}{\pgf@ya}}
		\pgfpathlineto{\pgfqpoint{\pgf@xb}{\pgf@ya}}
		\pgfpathlineto{\pgfqpoint{\pgf@xb}{\pgf@yb}}
		\pgfpathlineto{\pgfqpoint{\pgf@xa}{\pgf@yb}}
		\pgfpathclose
		\pgfpathmoveto{\pgfqpoint{\pgf@xb}{\pgf@ya}}
		\pgfpathlineto{\pgfpointadd{\pgfpoint{\pgf@xb}{\pgf@ya}}{\ppd@offset}}
		\pgfpathlineto{\pgfpointadd{\pgfpoint{\pgf@xb}{\pgf@yb}}{\ppd@offset}}
		\pgfpathlineto{\pgfpointadd{\pgfpoint{\pgf@xa}{\pgf@yb}}{\ppd@offset}}
		\pgfpathlineto{\pgfqpoint{\pgf@xa}{\pgf@yb}}
		\pgfpathmoveto{\pgfqpoint{\pgf@xb}{\pgf@yb}}
		\pgfpathlineto{\pgfpointadd{\pgfpoint{\pgf@xb}{\pgf@yb}}{\ppd@offset}}
	}
}
\makeatother

\tikzstyle{server}=[
parallelepiped,
fill=white, draw=white,
minimum width=0.35cm,
minimum height=0.75cm,
parallelepiped offset x=3mm,
parallelepiped offset y=2mm,
xscale=-1,
path picture={
	\draw[top color=white,bottom color=white]
	(path picture bounding box.south west) rectangle 
	(path picture bounding box.north east);
	\coordinate (A-center) at ($(path picture bounding box.center)!0!(path
	picture bounding box.south)$);
	\coordinate (A-west) at ([xshift=-0.575cm]path picture bounding box.west);
	\draw[ports]([yshift=0.1cm]$(A-west)!0!(A-center)$)
	rectangle +(0.2,0.065);
	\draw[ports]([yshift=0.01cm]$(A-west)!0.085!(A-center)$)
	rectangle +(0.15,0.05);
	\fill[white]([yshift=-0.35cm]$(A-west)!-0.1!(A-center)$)
	rectangle +(0.235,0.0175);
	\fill[white]([yshift=-0.385cm]$(A-west)!-0.1!(A-center)$)
	rectangle +(0.235,0.0175);
	\fill[white]([yshift=-0.42cm]$(A-west)!-0.1!(A-center)$)
	rectangle +(0.235,0.0175);
}
]

\tikzstyle{ports}=[
line width=0.3pt,
top color=white,
bottom color=white
]
\usepackage{pifont}
\newcommand*\colourcheck[1]{%
	\expandafter\newcommand\csname #1check\endcsname{\textcolor{#1}{\ding{52}}}%
}
\newcommand*\colourcross[1]{%
	\expandafter\newcommand\csname #1cross\endcsname{\textcolor{#1}{\ding{56}}}%
}
\colourcheck{istgreen}
\colourcross{istred}

\usepackage{pgfplots}
\pgfplotsset{compat=newest}
\pgfplotsset{plot coordinates/math parser=false}
\newlength\figureheight
\newlength\figurewidth 

\usepackage{amsmath}    %
\usepackage{amssymb}
\usepackage{amsfonts}    %
\usepackage{mathtools}
\usepackage{cuted}
\newcommand{\N}{\mathbb{N}} 

\newcommand{\Z}{\mathbb{Z}}

\newcommand{\Abc}{\boldsymbol{\mathcal{A}}}
\newcommand{\AbcTilde}{\boldsymbol{\mathcal{\tilde{A}}}}
\newcommand{\Bbc}{\boldsymbol{\mathcal{B}}}
\newcommand{\Cbc}{\boldsymbol{\mathcal{C}}}
\newcommand{\Dbc}{\boldsymbol{\mathcal{D}}}

\newcommand{\modq}{\bmod{q}}

\definecolor{istblue}{rgb/cmyk}{0.0,0.25490,0.56862745/1,0.55,0.0,0.43}%
\definecolor{istgreen}{rgb/cmyk}{0.388235,0.83137,0.4431372/0.53,0,0.46,0.16}%
\definecolor{istorange}{rgb/cmyk}{1.0,0.5333,0.0667/0.0,0.46,0.93,0.0}%
\definecolor{istred}{rgb/cmyk}{0.9961,0.2902,0.2863/0.0,0.7,0.71,0.0}%
\definecolor{istlightblue}{rgb/cmyk}{0.3765, 0.6863, 1.0/0.62,0.31,0.0}%
\definecolor{istdarkblue}{rgb/cmyk}{0.1176,0.1804,0.8706/0.86,0.79,0.0,0.12}%
\definecolor{istdarkgreen}{rgb/cmyk}{0.0627,0.5882,0.2824/0.89,0.0,0.52,0.41}%
\definecolor{istdarkred}{rgb/cmyk}{0.529,0.031,0.075/0,0.94,0.86,0.47}%
\definecolor{istlogoblue}{rgb/cmyk/gray}{0,0,0.804/1,1,0,0.2/0}%
\definecolor{unianthrazit}{rgb/cmyk}{0.24314,0.26667,0.29804/0.5,0.35,0.25,0.70}%
\definecolor{unimiddleblue}{rgb/cmyk}{0,0.31765,0.61961/1,0.7,0,0}%
\definecolor{unilightblue}{rgb/cmyk}{0,0.74510,1/0.7,0,0,0}%

\newcommand{\diagdots}[3][-25]{%
	\rotatebox{#1}{\makebox[0pt]{\makebox[#2]{\xleaders\hbox{$\cdot$\hskip#3}\hfill\kern0pt}}}%
}

\newtheorem{definition}{Definition}
\newtheorem{theorem}{Theorem}

\newtheorem{lemma}{Lemma}

%% file: BS-RS-FIR-C-A.bbl
\begin{thebibliography}{10}

\bibitem{Cheon2017}
Jung~Hee Cheon, Andrey Kim, Miran Kim, and Yongsoo Song.
\newblock Homomorphic {Encryption} for {Arithmetic} of {Approximate} {Numbers}.
\newblock In {\em Advances in {Cryptology} -- {ASIACRYPT} 2017}, pages
  409--437, 2017.

\bibitem{Kogiso2015}
Kiminao Kogiso and Takahiro Fujita.
\newblock Cyber-security enhancement of networked control systems using
  homomorphic encryption.
\newblock In {\em Proc. 54th {IEEE} Conf. Decision and Control ({CDC})}, 2015.

\bibitem{Alexandru}
Andreea~B Alexandru, Anastasios Tsiamis, and George~J Pappas.
\newblock Towards private data-driven control.
\newblock In {\em Proc. 59th IEEE Conf. Decision and Control (CDC)}, pages
  5449--5456, 2020.

\bibitem{Schluter2022a}
Nils Schl{\"u}ter, Matthias Neuhaus, and Moritz {Schulze Darup}.
\newblock Encrypted extremum seeking for privacy-preserving {PID} tuning
  as-a-service.
\newblock In {\em Proc. 2022 European Control Conf. ({ECC})}, 2022.

\bibitem{Schlueter2023}
Nils Schl{\"u}ter, Philipp Binfet, and Moritz Schulze~Darup.
\newblock A brief survey on encrypted control: From the first to the second
  generation and beyond.
\newblock {\em Annual Reviews in Control}, page 100913, 2023.

\bibitem{Murguia2020}
Carlos Murguia, Farhad Farokhi, and Iman Shames.
\newblock Secure and private implementation of dynamic controllers using
  semihomomorphic encryption.
\newblock {\em {IEEE} Trans. Autom. Control}, 65(9):3950--3957, 2020.

\bibitem{Schluter2021}
Nils Schl{\"u}ter, Matthias Neuhaus, and Moritz~Schulze Darup.
\newblock Encrypted dynamic control with unlimited operating time via {FIR}
  filters.
\newblock In {\em Proc. 2021 European Control Conf. ({ECC})}, pages 952--957,
  2021.

\bibitem{Adamek2024a}
Janis Adamek, Nils Schl{\"u}ter, and Moritz Schulze~Darup.
\newblock On the design of stabilizing {FIR} controllers.
\newblock In {\em Proc. 10th Int. Conf. Control, Decision and Information
  Technologies (CoDIT)}, pages 2037--2042, 2024.

\bibitem{gentry2009fully}
Craig Gentry.
\newblock {\em A fully homomorphic encryption scheme}.
\newblock Stanford University, 2009.

\bibitem{Cheon2018a}
Jung~Hee Cheon, Kyoohyung Han, Andrey Kim, Miran Kim, and Yongsoo Song.
\newblock Bootstrapping for approximate homomorphic encryption.
\newblock In {\em Advances in Cryptology -- {EUROCRYPT} 2018}, pages 360--384,
  2018.

\bibitem{Badawi2023}
Ahmad~Al Badawi and Yuriy Polyakov.
\newblock Demystifying bootstrapping in fully homomorphic encryption.
\newblock Cryptol. ePrint Arch., Paper 2023/149, 2023.

\bibitem{Marcolla2022a}
Chiara Marcolla, Victor Sucasas, Marc Manzano, Riccardo Bassoli, Frank H.~P.
  Fitzek, and Najwa Aaraj.
\newblock Survey on {Fully} {Homomorphic} {Encryption}, {Theory}, and
  {Applications}.
\newblock {\em Proc. IEEE}, 110(10):1572--1609, 2022.

\bibitem{schlor24a}
Sebastian Schlor and Frank Allg{\"o}wer.
\newblock Bootstrapping guarantees: Stability and performance analysis for
  dynamic encrypted control.
\newblock {\em IEEE Control Systems Letters}, 8:2235--2240, 2024.

\bibitem{Kim2016}
Junsoo Kim, Chanhwa Lee, Hyungbo Shim, Jung~Hee Cheon, Andrey Kim, Miran Kim,
  and Yongsoo Song.
\newblock Encrypting controller using fully homomorphic encryption for security
  of cyber-physical systems.
\newblock {\em {IFAC}-{PapersOnLine}}, 49(22):175--180, 2016.

\bibitem{Cheon2018b}
Jung~Hee Cheon, Kyoohyung Han, Hyuntae Kim, Junsoo Kim, and Hyungbo Shim.
\newblock Need for {Controllers} {Having} {Integer} {Coefficients} in
  {Homomorphically} {Encrypted} {Dynamic} {System}.
\newblock In {\em Proc. 57th IEEE Conf. Decision and Control (CDC)}, pages
  5020--5025, 2018.

\bibitem{Kim2021}
Junsoo Kim, Hyungbo Shim, Henrik Sandberg, and Karl~H. Johansson.
\newblock Method for running dynamic systems over encrypted data for infinite
  time horizon without bootstrapping and re-encryption.
\newblock In {\em Proc. 60th IEEE Conf. Decision and Control (CDC)}, 2021.

\bibitem{Lee2023b}
Joowon Lee, Donggil Lee, Seungbeom Lee, Junsoo Kim, and Hyungbo Shim.
\newblock Conversion of controllers to have integer state matrix for encrypted
  control: Non-minimal order approach.
\newblock In {\em Proc. 62nd IEEE Conf. Decision and Control (CDC)}, pages
  5091--5096, 2023.

\bibitem{Lee2025}
Joowon Lee, Donggil Lee, and Junsoo Kim.
\newblock Stabilization by controllers having integer coefficients, 2025.
\newblock Preprint: https://arxiv.org/abs/2505.00481.

\bibitem{scherer2000linear}
Carsten Scherer and Siep Weiland.
\newblock Linear matrix inequalities in control.
\newblock {\em Lecture Notes, Dutch Institute for Systems and Control, Delft,
  The Netherlands}, 3(2), 2000.

\bibitem{Lang2024}
Simon Lang, Marc Seidel, and Frank Allgöwer.
\newblock Robust performance for switched systems with constrained switching
  and its application to weakly hard real-time control systems.
\newblock In {\em Cyber-physical Networking (in publication)}. Springer, 2024.
\newblock Preprint: https://arxiv.org/abs/2411.08436v1.

\bibitem{Green1995}
Michael Green and David J.~N. Limebeer.
\newblock {\em Linear robust control}.
\newblock Prentice-Hall, Englewood Cliffs, NJ, 1995.

\bibitem{mosek}
{MOSEK ApS}.
\newblock {\em MOSEK Optimization Toolbox for MATLAB 10.2.17}, 2024.

\bibitem{Lofberg2004}
J.~L{\"{o}}fberg.
\newblock {YALMIP} : A toolbox for modeling and optimization in {MATLAB}.
\newblock In {\em In Proceedings of the CACSD Conference}, Taipei, Taiwan,
  2004.

\end{thebibliography}
